\documentclass[11pt]{article}

\usepackage[margin=1in]{geometry}

\usepackage{authblk}
\usepackage[hypertexnames=false,colorlinks=true,urlcolor=blue,citecolor=blue,linkcolor=black]{hyperref}\usepackage{todonotes}

\usepackage[numbers,sort&compress]{natbib}

\usepackage{xspace}
\usepackage[T1]{fontenc}
\usepackage[utf8]{inputenc}
\usepackage{xcolor}
\usepackage{amsmath}
\usepackage{amssymb}
\usepackage{amsthm}
\usepackage{thmtools}
\usepackage{makecell}
\usepackage{ifthen}
\usepackage[capitalize]{cleveref}
\usepackage{enumerate}
\usepackage{tikz}
\usetikzlibrary{calc} 
\usetikzlibrary{decorations.pathreplacing} 
\usetikzlibrary{patterns}
\usetikzlibrary{math} 

\usepackage{booktabs}
\usepackage{algorithm}
\usepackage[noend]{algpseudocode}
\usepackage{xcolor,stackengine,graphicx,scalerel,xspace}
\usepackage{multirow}
\usepackage{caption}

\newcommand{\wild}{\protect\scalebox{.7}{\protect\stackinset{c}{}{c}{}{$\lozenge$}{\textcolor{black!10}{\raisebox{.2em}{$\blacklozenge$}}}}}
\def\dd{\mathinner{.\,.}}
\newcommand{\eps}{\varepsilon}
\newcommand{\cO}{\mathcal{O}}
\newcommand{\sfup}[1]{\textsf{\textup{#1}}}
\newcommand{\LCE}{\sfup{LCE}\xspace}
\newcommand{\LCEW}{\sfup{LCEW}\xspace}
\newcommand{\kpme}{$k$-\sfup{PME}\xspace}
\newcommand{\kpmwe}{$k$-\sfup{PMWE}\xspace}
\newcommand{\cOtilde}{\tilde{\cO}}
\newcommand{\ceil}[1]{\lceil #1 \rceil}
\newcommand{\Ceil}[1]{\left\lceil #1 \right\rceil}
\newcommand{\absolute}[1]{\lvert #1 \rvert}

\newcommand{\sd}{\textsf{Set-Disjointness}\xspace}
\newcommand{\tsum}{\textsf{3SUM}\xspace}
\newcommand{\sdconjstrong}{\textsf{Strong \sd Conjecture}\xspace}
\newcommand{\tsumconj}{\textsf{\tsum Conjecture}\xspace}

\DeclareMathOperator{\polylog}{polylog}

\newcommand{\ed}{\mathrm{ed}}

\newcommand{\jump}{\textsc{Jump}\xspace}
\newcommand{\Ss}{\mathsf{Sel}}
\newcommand{\Tt}{\mathsf{Tr}}
\newcommand{\nextpos}{\textsf{next\_tr}\xspace} 
\newcommand{\nextsel}{\textsf{next\_sel}\xspace} 
\newcommand{\IN}{m_{in}}
\newcommand{\OUT}{m_{out}} 

\newcommand{\matches}{\sim{}}
\newcommand{\mismatches}{\nsim{}}

\newlength{\problemparwidth}
\newcommand{\problemdef}[5]{%
\begin{center}
\noindent\fbox{%
  \begin{minipage}{0.982\linewidth}
    \textbf{\textsf{#1}}\\
    {\bf{#2:}} #3\\
    {\bf{#4:}} #5
  \end{minipage}%
  }
  \end{center}
}

\newcommand{\problemoutput}[3]{\problemdef{#1}{Input}{#2}{Output}{#3}}
\newcommand{\problemtask}[3]{\problemdef{#1}{Input}{#2}{Task}{#3}}
\newcommand{\problemquery}[3]{\problemdef{#1}{Input}{#2}{Query}{#3}}

\newtheorem{theorem}{Theorem}
\newtheorem{lemma}[theorem]{Lemma}
\newtheorem{corollary}[theorem]{Corollary}
\newtheorem{conjecture}[theorem]{Conjecture}
\newtheorem{fact}[theorem]{Fact}

\newenvironment{claimproof}{\begin{proof}}{\end{proof}}

\newtheorem{claim}[theorem]{Claim}%
\newtheorem{example}[theorem]{Example}%
\newtheorem{observation}[theorem]{Observation}%

\title{\vspace{-2ex}Longest Common Extensions with Wildcards: Trade-off and Applications}

\author[1,2]{Gabriel Bathie}

\author[3,4]{Itai Boneh}

\author[5]{Panagiotis Charalampopoulos}

\author[1]{Jonas Ellert}

\author[1]{Tatiana Starikovskaya}

\affil[1]{DIENS, \'{E}cole Normale Sup\'{e}rieure, Paris, France\\
\texttt{$\{$gabriel.bathie,ellert.jonas,tat.starikovskaya$\}$@gmail.com}}

\affil[2]{LaBRI, Universit\'{e} de Bordeaux, France}

\affil[3]{Reichman University, Herzliya, Israel\\
    \texttt{itai.bone@biu.ac.il}}

\affil[4]{University of Haifa, Israel}

\affil[5]{Birkbeck, University of London, United Kingdom\\
    \texttt{pcharalampo@gmail.com}}
\date{\vspace{-5ex}}

\begin{document}

\maketitle

\thispagestyle{empty}

\abstract{%
  We study the Longest Common Extension (LCE)
  problem in a string containing wildcards.
  Wildcards (also called ``don't cares'' or ``holes'')
  are special characters that match any other character in the alphabet,
  similar to the character ``\texttt{?}'' in Unix commands or ``\texttt{.}'' in regular expression engines.

  We consider the problem parametrized by $G$, the number of maximal
  contiguous groups of wildcards in the input string.
  Our main contribution is a simple data structure for this problem that can be built in
  $\mathcal O(n (G/t) \log n)$ time, occupies $\mathcal O(nG/t)$ space,
  and answers queries in $\mathcal O(t)$ time, for any $t \in [1 \dd G]$.
  Up to the $\mathcal O(\log n)$ factor, this interpolates smoothly between the 
  data structure of Crochemore et al.~[JDA 2015], which has $\mathcal O(nG)$ preprocessing time and space, and $\mathcal O(1)$ query time,
  and a simple solution based on the ``kangaroo jumping'' technique [Landau and Vishkin, STOC 1986],
  which has $\mathcal O(n)$ preprocessing time and space, and $\mathcal O(G)$ query time. 
  
  By establishing a connection between this problem and Boolean matrix multiplication, we show that our solution is optimal, up to subpolynomial factors, among combinatorial data structures when $G = \Omega(n^\eps)$ under a widely believed hypothesis.
  In addition, we develop a simple deterministic combinatorial algorithm for sparse Boolean matrix multiplication. We further establish a conditional lower bound for non-combinatorial data structures, stating that $\mathcal O(nG/t^4)$ preprocessing time (resp.\ space) is optimal, up to subpolynomial factors, for any data structure with query time $t$ for a wide range of $t$ and $G$, assuming the well-established \tsum (resp.\ \sd) conjecture.
  
  Finally, we show that our data structure can be used to obtain
  efficient algorithms for approximate pattern matching
  and structural analysis of strings with wildcards.
  First, we consider the problem of pattern matching with $k$ errors (i.e., edit operations) in the setting where both the pattern and the text may contain wildcards.
  The ``kangaroo jumping'' technique can be adapted to yield an algorithm for this problem with time complexity $\mathcal O(n(k+G))$, 
  where $G$ is the total number of maximal contiguous groups of wildcards in the text and the pattern and $n$ is the length of the text.
  By combining ``kangaroo jumping'' with a tailor-made data structure for LCE queries, Akutsu [IPL 1995] devised an $\mathcal O(n\sqrt{km} \polylog m)$-time algorithm.
  We improve on both algorithms when $k \ll G \ll m$ by giving an algorithm running in time $\mathcal O(n(k + \sqrt{Gk \log n}))$. 
Secondly, we give $\mathcal O(n\sqrt{G} \log n)$-time and $\mathcal O(n)$-space algorithms for computing the prefix array, as well as the quantum/deterministic border and period arrays of a string with wildcards.
This is an improvement over the $\mathcal O(n\sqrt{n\log n})$-time algorithms of Iliopoulos and Radoszewski [CPM 2016] when $G =  o(n / \log n)$.  
}
  
\clearpage
\setcounter{page}{1}

\section{Introduction}
Given a string $T$, the \emph{longest common extension} (\LCE) at indices $i$ and $j$
is the length of the longest common prefix of the suffixes of $T$ starting at indices $i$ and $j$.
In the \LCE problem, given a string $T$, the goal is to build a data structure that can efficiently answer \LCE queries.

Longest common extension queries are a powerful string operation
that underlies a myriad of string algorithms, for problems such as
approximate pattern matching~\cite{akutsu, amir2004faster, unified, galil1986improved, landau1986efficient, kangaroo},
finding maximal or gapped palindromes~\cite{LCE-tradeoffs, DBLP:conf/cpm/Charalampopoulos22, DBLP:books/cu/Gusfield1997, kolpakov2009searching}, and computing the repetitive structure (e.g., runs) in strings~\cite{bannai2014new, kolpakov1999finding}, to name just a few.

Due to its importance, the \LCE problem and its variants have received a lot of attention~\cite{
bille2018finger,
bille2017fingerprints,
bille2015longest,
LCE-tradeoffs,
10.5555/3381089.3381126,
10.1007/11780441_5,
DBLP:conf/soda/GawrychowskiK17,
doi:10.1137/1.9781611975031.99,
doi:10.1137/0213024,
DBLP:conf/stoc/KempaK19,
KociumakaPhD,
10353220,
DBLP:journals/corr/abs-2105-03782,
nishimoto2016fully,
10.1145/3426870,
tanimura2016deterministic,
tanimura2017small,
tomohiro2017longest,doi:10.1137/1.9781611977073.111}.
The suffix tree of a string of length~$n$ occupies $\Theta(n)$ space and can be preprocessed in $\cO(n)$ time to answer \LCE queries in constant time~\cite{10.1007/11780441_5,doi:10.1137/0213024}. 
However, the $\Theta(n)$ space requirement can be prohibitive for applications such as computational biology that deal with extremely large strings. 
Consequently, much of the recent research has focused on designing data structures that use less space without being (much) slower in answering queries.
Consider the setting when we are given a read-only length-$n$ string $T$ over an alphabet of size polynomial in $n$.
Bille et al.~\cite{bille2015longest} gave a data structure for the \LCE problem that, for any given user-defined parameter $\tau \le n$, occupies $\cO(\tau)$ space on top of the input string and answers queries in $\cO(n/\tau)$ time.
Kosolobov~\cite{kosolobov2017tight} showed that this data structure is optimal when $\tau = \Omega(n/ \log n)$.
A drawback of the data structure of Bille et al.~\cite{bille2015longest} is its rather slow $\cO(n^{2+\varepsilon})$ construction time.
This motivated studies towards an \LCE data structure with optimal space and query time and a fast construction algorithm.
Gawrychowski and Kociumaka~\cite{DBLP:conf/soda/GawrychowskiK17} gave an optimal $\cO(n)$-time and $\cO(\tau)$-space Monte Carlo construction algorithm and Birenzwige et al.~\cite{10.5555/3381089.3381126} gave a Las Vegas construction algorithm with the same complexity provided $\tau = \Omega(\log n)$.
Finally, Kosolobov and Sivukhin~\cite{DBLP:journals/corr/abs-2105-03782} gave a deterministic construction algorithm that works in optimal $\cO(n)$ time and $\cO(\tau)$ space for $\tau = \Omega(n^\varepsilon)$, where $\varepsilon > 0$ is an arbitrary constant.
Another line of work~\cite{bille2018finger,bille2017fingerprints,doi:10.1137/1.9781611975031.99,nishimoto2016fully, 
tanimura2016deterministic, tanimura2017small,tomohiro2017longest,10353220,doi:10.1137/1.9781611977073.111} considers \LCE data structures over compressed strings.

One important variant of the \LCE problem is that of \LCE with $k$-mismatches ($k$-\LCE),
where one wants to find the longest prefixes that differ in at most $k$ positions, for a given integer parameter $k$.
Landau and Vishkin~\cite{kangaroo} proposed a technique, dubbed ``kangaroo jumping'', that reduces $k$-\LCE to $k+1$ standard \LCE queries. This technique is a central component of many approximate pattern matching algorithms, under the Hamming~\cite{amir2004faster, unified} and the edit~\cite{akutsu, unified, kangaroo} distances.

In this work, we focus on the variant of \LCE in strings with \emph{wildcards}, denoted \LCEW. Wildcards (also known as \emph{holes} or \emph{don't cares}), denoted $\wild$, are special characters that match every character of the alphabet. Wildcards are a versatile tool for modeling uncertain data, and algorithms on strings with wildcards have garnered considerable attention in the literature~\cite{clifford2007simple,
DBLP:conf/stoc/ColeH02,
FP74,DBLP:conf/focs/Indyk98a,
DBLP:conf/soda/Kalai02a,
DBLP:conf/stacs/CliffordGLS18,
DBLP:journals/algorithmica/GolanKP19,
DBLP:conf/soda/Fischer24,
akutsu,
DBLP:conf/soda/CliffordEPR09,DBLP:journals/jcss/CliffordEPR10,DBLP:journals/algorithms/NicolaeR15,amir2004faster,DBLP:journals/ipl/NicolaeR17,DBLP:journals/corr/abs-2402-07732,
DBLP:conf/stoc/ColeGL04,DBLP:journals/siamcomp/Patrascu11,DBLP:journals/mst/BilleGVV14,DBLP:conf/icalp/AfshaniN16,
DBLP:conf/icalp/AbboudWW14,
iliopoulos2016truly,
DBLP:journals/dam/ManeaMT14,
DBLP:conf/iwoca/Blanchet-SadriH15,
DBLP:journals/dam/Blanchet-SadriO18,
DBLP:journals/ejc/BollobasL18}. 

Given a string $T$, and indices $i,j$, $\LCEW(i,j)$ is the length of the longest matching prefixes of the suffixes of $T$ starting at indices $i$ and $j$. For all $\tau \in [1 \dd n]$, Iliopoulos and Radoszewski~\cite{iliopoulos2016truly} showed an \LCEW data structure with $\cO(n^2 \log n / \tau)$ preprocessing time, $\cO(n^2/\tau)$ space, and $\cO(\tau)$ query time.
In the case where the number of wildcards in~$T$ is bounded, more efficient data structures exist. 
The \LCEW problem is closely related to $k$-\LCE:
if we let $\widehat{T}$ be the string obtained by replacing each wildcard in~$T$
with a new character, the $i$-th wildcard replaced with a fresh letter $\#_i$, then 
an \LCEW query in $T$ can be reduced to a $D$-\LCE query in $\widehat{T}$,
where $D$ is the number of wildcards in $T$. Consequently, an \LCEW query can be answered using $\cO(D)$ \LCE queries.
In particular, if we use the suffix tree to answer \LCE queries, we obtain a data structure with $\cO(n)$ space and construction time and $\cO(D)$ query time. At the other end of the spectrum, Blanchet-Sadri and Lazarow~\cite{10.1007/978-3-642-37064-9_16} showed that one can achieve $\cO(1)$ query time using $\cO(nD)$ space after an $\cO(nD)$-time preprocessing.
\begin{example}\label{example:param}
For string $T = \texttt{abab} \wild \wild \wild \texttt{aaaa} \wild \wild \wild \wild \texttt{ba} \wild \wild \wild \texttt{bb}$, we have $D = 10$ and $G = 3$.
\end{example}
By using the structure of the wildcards inside the string, one can improve the aforementioned bounds even further.
Namely, it is not hard to see that if the wildcards in $T$ are arranged in~$G$ maximal contiguous groups (see Example~\ref{example:param}), then we can reduce the number of \LCE queries needed to answer an \LCEW query to $G$ by jumping over such groups, thus obtaining a data structure with $\cO(n)$-time preprocessing, $\cO(n)$ space, and $\cO(G)$ query time.
On the other hand, Crochemore et al.~\cite{crochemore2015note} devised an $\cO(nG)$-space data structure that can be built in $\cO(nG)$ time and can answer \LCEW queries in constant time.

\subsection{Our Results}
In this work, we present an \LCEW data structure
that achieves a smooth space-time trade-off between the data structure based on ``kangaroo jumps'' and that of Crochemore et al.~\cite{crochemore2015note}.
As our main contribution, we show that for any $t\le G$, there exists a set of $\cO(G/t)$
positions, called \textit{selected} positions, 
that intersects any chain of $t$ kangaroo jumps from a fixed pair of positions.
Given the \LCEW information on selected positions,
we can speed up \LCEW queries on arbitrary positions
by jumping from the first selected position in the longest common extension 
to the last selected position in the longest common extension.
This gives us an $\cO(t)$ bound on the number of kangaroo jumps we need to perform to answer a query.
We leverage the fast FFT-based algorithm of Clifford and Clifford~\cite{clifford2007simple}
for pattern matching with wildcards to efficiently
build a dynamic programming table containing the result
of \LCEW queries on pairs of indices containing a selected position;
this table allows us to jump from the first to the last selected position in the longest common extension in constant time.
The size of the table is $\cO(nG/t)$, while the query time is $\cO(t)$.
For comparison, the data structures of Crochemore et al.~\cite{crochemore2015note}
and of Iliopoulos and Radoszewski~\cite{iliopoulos2016truly} use a similar dynamic programming scheme
that precomputes the result of \LCEW queries for a subset of positions:
Crochemore et al.~use all transition positions (see \cref{sec:data-structure} for a definition),
while Iliopoulos and Radoszewski use one in every $\sqrt{n}$ positions.
We use a more refined approach, that allows us to obtain both a dependency on $G$ instead of $n$ and a space-query-time trade-off.
Our result can be stated formally as follows.

\begin{restatable}{theorem}{maintheorem}\label{thm:data-structure}
  Suppose that we are given a string $T$ of length $n$ that contains wildcards arranged into $G$ maximal contiguous groups.
  For every $t \in [1 \dd G]$, there exists a deterministic data structure that:
  \begin{itemize}
    \item uses space $\cO(nG/t)$, 
    \item can be built in time $\cO(n(G/t)\log n)$ using $\cO(nG/t)$ space,
    \item given two indices $i,j \in [1\dd n]$, returns $\LCEW(i,j)$ in time $\cO(t)$.
  \end{itemize}

\end{restatable}

We further show that this trade-off can be extended to $t \ge G$ by implementing the kangaroo jumping method of Landau and Vishkin~\cite{kangaroo} with a data structure that provides a time-space trade-off for (classical) \LCE queries. Using the main result of Kosolobov and Sivukhin~\cite{DBLP:journals/corr/abs-2105-03782}, we obtain the following:

\begin{corollary}\label{cor:extended-tradeoff}
Suppose that we are given a read-only string $T$ of length $n$ that contains wildcards arranged into $G$ maximal contiguous groups.
For every constant $\varepsilon > 0$ and $t \in [G \dd G \cdot n^{1-\varepsilon}]$, there exists a data structure that:
\begin{itemize}
\item uses space $\cO(nG/t)$,
\item can be built in time $\cO(n)$ using $\cO(nG/t)$ space, 
\item given two indices $i,j \in [1\dd n]$, returns $\LCEW(i,j)$ in time $\cO(t)$.
\end{itemize}
\end{corollary}
\begin{proof}
We build the \LCE data structure of Kosolobov and Sivukhin~\cite{DBLP:journals/corr/abs-2105-03782} for parameter $\tau = nG/t = \Omega(n^\varepsilon)$ in $\cO(n)$ time and $\cO(\tau)$ space.
As an \LCEW query reduces to $G$ \LCE queries, the constructed data structure supports \LCEW queries in $\cO(G \cdot n/\tau) = \cO(t)$ time. 
\end{proof}

By a reduction from Boolean matrix multiplication,
we derive a conditional $\Omega((nG)^{1-o(1)})$ lower bound on the product of the preprocessing and query times of any combinatorial data structure for the \LCEW problem for values of $G$ polynomial in $n$ (\cref{thm:lower-bound}).\footnote{In line with previous work, we say that an algorithm or a data structure is \textit{combinatorial} if it does not
use fast matrix multiplication as a subroutine during preprocessing
or while answering queries.}
This lower bound matches the trade-off of our data structure up to subpolynomial factors.

\begin{table}[h!]
  \centering
  \captionsetup{width=.88\linewidth}
\setlength\extrarowheight{1ex}
\caption{Overview of combinatorial deterministic sparse Boolean matrix multiplication algorithms. The values $\IN$ (resp. $\OUT$) refer to the total number of non-zero entries in the input matrices (resp., in the output matrix).} 
\small
\begin{tabular*}{.85\textwidth}{p{.58\textwidth}p{.38\textwidth}}
    \toprule
    Source & Running Time \\[.5ex]
    \midrule
    Gustavson~\cite{Gustavson78} & $\cO(n\cdot \IN)$ \\
    Kutzkov~\cite{Kutzkov13} & $\cO(n\cdot (n+ \OUT^2))$\\ 
    Künnemann~\cite{Kunnemann18} & $\cO(\sqrt{\OUT} \cdot n^2 + \OUT^2)$ \\
    Abboud, Bringmann, Fischer, Künemann~\cite{abboud2023time} & $\cOtilde(\IN\sqrt{\OUT})$ \\
    \textbf{Our algorithm} & $\cOtilde(n \sqrt{\IN \cdot (n + \OUT)})$\\
    \bottomrule
\end{tabular*}
\label{tab:fmm}
\end{table}

Surprisingly, one can also use the connection between the two problems to derive an algorithm for sparse Boolean matrix multiplication (BMM). Existing algorithms for BMM can be largely categorised into two types: combinatorial and those relying on (dense) fast matrix multiplication. However, the latter are notorious for their significant hidden constants, making them unlikely candidates for practical applicability. By using the connection to the \LCEW problem, we show a deterministic combinatorial algorithm with time complexity $\cOtilde(n \sqrt{\IN \cdot (n + \OUT)})$. Our algorithm ties or outperforms all other known deterministic combinatorial algorithms~\cite{Gustavson78,Kutzkov13,GuchtWWZ15,Kunnemann18}
for some range of parameters $\IN$ and $\OUT$, e.g., when $\IN = \Theta(n^{3/2})$ and $\OUT = \Theta(n^{4/3})$, except for the one implicitly implied by the result of Abboud et al.~\cite{abboud2023time}.
Namely, by replacing fast matrix multiplication (used in a black-box way) in~\cite[Theorem 4.1]{abboud2023time}  with the naive matrix multiplication algorithm, one obtains a deterministic combinatorial algorithm with running time $\cOtilde(\IN \sqrt{\OUT})$, which is always better than our time bound. See \cref{tab:fmm} for a summary. However, our algorithm is much simpler than that of Abboud et al.~\cite{abboud2023time}: while our algorithm relies solely on standard tools typically covered in undergraduate computer science courses, theirs requires an intricate construction of a family of hash functions with subsequent derandomisation. We provide a (non-optimized) proof-of-concept implementation at \url{https://github.com/GBathie/LCEW}.

Finally, we establish a reduction from Set Disjointness to \LCEW using a technique similar to that of Kopelowitz and Vassilevska Williams~\cite{DBLP:conf/icalp/KopelowitzW20}. As a direct consequence, we prove that any--whether combinatorial or not--\LCEW data structure for a string of length $n$ with $G$ groups of wildcards with query time $t$ cannot use $\cO(nG/t^4)$ space and preprocessing time, assuming widely accepted hardness hypotheses (\cref{thm:3sum_and_sd_conditional_lower}).

In conclusion, the three lower bounds (the bound for combinatorial structures from \cref{thm:lower-bound} conditioned on the hardness of combinatorial matrix multiplication and the conditional bounds from \cref{thm:3sum_and_sd_conditional_lower}) strongly suggest that significantly improving our solution for the \LCEW problem is highly unlikely, especially for small values of $t$. We leave as an intriguing open question whether a tight non-combinatorial lower bound exists, which would further solidify our understanding of the problem's complexity.

\paragraph*{Applications.} 
We further showcase the significance of our data structure by using it to improve over the state-of-the-art algorithms for approximate pattern matching and for the construction of periodicity-related arrays for strings containing wildcards. 

As previously mentioned, \LCE queries play a crucial role in string algorithms, especially in approximate pattern matching algorithms, such as for the problem of \emph{pattern matching with $k$ errors} (\kpme, also known as pattern matching with $k$ \emph{edits} or \textit{differences}). This problem involves identifying all positions in a given text where a fragment starting at that position is within edit distance $k$ from a given pattern. The now-classical Landau--Vishkin algorithm~\cite{kangaroo} elegantly solves this problem, achieving a time complexity of $\cO(nk)$ through extensive use of \LCE queries. A natural extension of \kpme is the problem of \textit{pattern matching with wildcards and $k$-errors} (\kpmwe), where the pattern and the text may contain wildcards. The algorithm of Landau and Vishkin~\cite{kangaroo} for \kpme can be extended to an $\cO(n(k+G))$-time algorithm for \kpmwe in strings with $G$ groups of wildcards (see~\cite{akutsu}). Building on~\cite{kangaroo}, Akutsu~\cite{akutsu} gave an algorithm for \kpmwe that runs in time $\cOtilde(n\sqrt{km})$.\footnote{Throughout this work, the $\tilde{O}(\cdot)$ notation suppresses factors that are polylogarithmic in the total length of the input string(s).} In \cref{thm:alg-pm}, we give an algorithm for \kpmwe with time complexity $O(n(k + \sqrt{Gk \log n}))$, which improves on the algorithms of Akutsu~\cite{akutsu} and Landau and Vishkin~\cite{kangaroo} in the regime where $k \ll G \ll m$.

Periodicity arrays capture repetitions in strings and are widely used in pattern matching algorithms; for instance, see~\cite{DBLP:books/daglib/0020111,DBLP:books/daglib/0020103}.  
The prefix array of a length-$n$ string~$T$ with wildcards stores $\LCEW(1,j)$ for all $1 \le j \le n$.
It was first studied in~\cite{DBLP:conf/stringology/IliopoulosMMPST02}, where an $\cO(n^2)$-time construction algorithm was given. More recently, Iliopoulos and Radoszewski~\cite{iliopoulos2016truly} presented an $O(n\sqrt{n\log n})$-time and $\Theta(n)$-space algorithm.
Another fundamental periodicity array is the border array, which stores the maximum length of a proper border of each prefix of the string.
When a string contains wildcards, borders can be defined in two different ways~\cite{borders,DBLP:conf/stringology/IliopoulosMMPST02}.
A \emph{quantum border} of a string~$T$ is a prefix of $T$ that matches the same-length suffix of $T$, while a \emph{deterministic border} is a border of a string $T'$ that does not contain wildcards and matches~$T$, see~\cref{example:borders}. A closely related notion is that of quantum and deterministic periods and their respective period arrays (see~\cref{sec:prelim} for definitions). 

\begin{example}\label{example:borders}
The maximal length of a quantum border of $T = \texttt{ab} \wild \texttt{bc}$ is 3; note that $\texttt{ab} \wild$ matches $\wild \texttt{bc}$.
The maximal length of a deterministic border of $T$, however, is $0$.
\end{example}

Early work in this area \cite{borders,DBLP:conf/stringology/IliopoulosMMPST02} showed that both variants of the border array can be constructed in $\cO(n^2)$ time.
Iliopoulos and Radoszewski~\cite{iliopoulos2016truly} demonstrated that one can compute the border arrays from the prefix array in $\cO(n)$ time and $\cO(n)$ space, and the period arrays in $\cO(n \log n)$ time and $\cO(n)$ space, thus deriving an $\cO(n \sqrt{n \log n})$-time, $\cO(n)$-space construction algorithm for all four arrays. 
In \cref{thm:arrays}, we give $O(n\sqrt{G} \log n)$-time, $O(n)$-space algorithms for computing the prefix array, as well as the quantum and  deterministic
border and period arrays, improving all previously known algorithms when $G = o(n / \log n)$. 

\section{Preliminaries}
\label{sec:prelim}
A string $S$ of length $n = |S|$ is a finite sequence of $n$ characters
over a finite alphabet~$\Sigma$. The $i$-th character of $S$ is denoted by $S[i]$, for $1\le i \le n$,
and we use $S[i\dd j]$ to denote the \emph{fragment} $S[i]S[i+1]\ldots S[j]$ of $S$
(if $i > j$, then $S[i\dd j]$ is the empty string).
Moreover, we use $S[i\dd j)$ to denote the fragment $S[i\dd j-1]$ of $S$.
A fragment $S[i\dd j]$ is a \textit{prefix} of $S$ if $i =1$ and a \emph{suffix} of $S$ if $j = n$.

In this paper, the alphabet $\Sigma$ contains a special character
$\wild$ that matches every character in the alphabet.
Formally, we define the ``match'' relation,
denoted $\matches $ and defined over $\Sigma\times\Sigma$, as follows: $\forall a,b\in\Sigma : a\matches  b \Leftrightarrow a = b \lor a = \wild \lor b = \wild$. 
Its negation is denoted $a \mismatches b$.
We extend this relation to strings of equal length
by $X\matches Y \Leftrightarrow \forall i=1, \ldots, |X|:  X[i]\matches Y[i]$.

\subparagraph*{Longest common extensions.}
Let $T$ be a string of length $n$, and let $i,j \le n$ be indices.
The longest common extension at $i$ and $j$ in $T$, denoted $\LCE_T(i,j)$
is defined as $\LCE_T(i,j) = \max\{\ell \le \min(n-i, n-j) + 1 : T[i\dd i+\ell) = T[j\dd j+\ell)\}$.
Similarly, the longest common extension \textit{with wildcards} is defined using the $\matches$ relation instead of equality: $\LCEW_T(i,j) = \max\{\ell \le \min(n-i, n-j) + 1: T[i\dd i+\ell) \matches T[j\dd j+\ell)\}$.

We focus on data structures for \LCEW queries inside a string $T$, but our results can easily be extended to answer queries between two strings $P,Q$, denoted $\LCEW_{P,Q}(i,j)$.
If we consider $T = P\cdot Q$,
then for any $i\le |P|$ and $j\le|Q|$, 
we have $\LCEW_{P,Q}(i,j) = \min(\LCEW_{T}(i, j+|P|), |P|-i+1, |Q|-j+1)$.
When the string(s) that we query are clear from the context, we drop the $T$ or $P,Q$ subscripts.

\subparagraph*{Periodicity arrays.} The \emph{prefix array} of a string $S$ of length $n$ is an array $\pi$ of size $n$ such that $\pi[i] = \LCEW(1, i)$.

An integer $b \in[1\dd n]$ is a \textit{quantum border} of $S$ if $S[1\dd b] \matches S[n-b+1 \dd n]$. It is a \textit{deterministic border} of $S$ if there exists a string $X$ \textit{without wildcards} such that $X\matches S$ and $X[1\dd b] = X[n-b+1\dd n]$.
Similarly, an integer $p \le n$ is a \textit{quantum period} of $S$ if for every $i\le n-p, S[i] \matches S[i+p]$, and it is a \textit{deterministic period} of $S$ if there exists a string $X$ \textit{without wildcards} such that $X\matches S$ and for every $i\le n-p, X[i] = X[i+p]$.

\begin{example}
  Consider string $\texttt{ab}\wild \texttt{b}\wild \texttt{bcb}$. Its smallest quantum period is $2$,
  while its smallest deterministic period is $4$.
\end{example}

\noindent For a string $S$ of length $n$, we define the following arrays of length $n$:
\begin{itemize}
  \item the period array $\pi$, where $\pi[i] = \LCEW(1, i)$;
  \item the deterministic and quantum border arrays, $B$ and $B_Q$, where $B[i]$ and $B_Q[i]$ are the largest deterministic and quantum border of $S[1\dd i]$, respectively;
  \item the deterministic and quantum period arrays, $P$ and $P_Q$, such that $P[i]$ and $P_Q[i]$ are the smallest deterministic and quantum periods of $S[1\dd i]$, respectively.
\end{itemize}

\begin{fact}[{Lemmas 12 and 15~\cite{iliopoulos2016truly}}]\label{fact:pref-border}
Given the prefix array of a string $S$,
one can compute the quantum border array and quantum period array
in $\cO(n)$ time and space,
while the deterministic border and period arrays can be computed
in $\cO(n\log n)$ time and $\cO(n)$ space.
\end{fact}

\paragraph*{Model of computation.} We work in the standard word RAM model of computation with word size $\Theta(\log N)$, where~$N$ is the size of the input.

\section{Time-Space Trade-off for \LCEW}\label{sec:data-structure}
In this section, we prove~\cref{thm:data-structure}.
Recall that $1 \le t \le G$. 
Following the work of Crochemore et al~\cite{crochemore2015note},
we define \textit{transition positions} in $T$,
which are the positions at which $T$ transitions
from a block of wildcards to a block of non-wildcards characters.
We use $\Tt$ to denote the set of transition positions in $T$.
Formally, a position $i\in [1\dd n]$ is in $\Tt$ if one of the following
holds:
\begin{itemize}
  \item $i = n$,
  \item $i > 1$, $T[i-1] = \wild$ and $T[i] \neq \wild$.
\end{itemize}
Note that as $T$ contains $G$ groups of wildcards, there are at most $G+1$
transition positions, i.e., $|\Tt| = \cO(G)$.
Moreover, by definition, the only transition position $i$ for which $T[i]$ may be a wildcard is $n$.

Our algorithm precomputes the \LCEW information for a subset of evenly distributed
transition positions, called \textit{selected positions} and denoted $\Ss$, 
whose number depends on the parameter~$t$.
The set $\Ss$ contains one in every $t$ transition position in $\Tt$,
along with the last one (which is $n$).
Formally, let $i_1 < i_2 < \ldots < i_r$ denote the transition positions of $T$,
sorted in increasing order,
then $\Ss = \{i_{st+1}: s = 0,\ldots, \lfloor (r-1)/t \rfloor \} \cup \{n\}.$
Let $\lambda$ denote the cardinality of $\Ss$, which is $\cO(G/t)$.

Additionally, for every $i\in[1\dd n]$, we define $\nextpos[i]$ (resp., $\nextsel[i]$) as the distance between $i$ and the next transition position (resp., the next selected position) in $T$.
Formally, $\nextpos[i] = \min\{j-i: j\in\Tt \land j \ge i\}$ and $\nextsel[i] = \min\{j-i: j\in\Ss \land j \ge i\}$. 
These values are well-defined: as $n$ is both a transition and a selected position,
the minimum in the above equations is never taken over the empty set.
Both arrays can be computed in linear time and stored using $\cO(n)$ space.
The array $\nextpos$ can be used to jump from a wildcard to the end of the group of wildcards containing it, due the following property:
\begin{observation}\label{fact:match-to-nextpos}
  Consider some position $i$ such that $T[i] =\wild$.
  Then, for $r = \nextpos[i]$,
  we have $T[i\dd i+r] = \wild^ra$, where $a\in \Sigma \setminus \{\wild\}$, i.e., the fragment from $i$ until the next transition position (exclusive) contains only wildcards.
\end{observation}

The central component of our data structure is a dynamic programming table,
\jump, which allows us to efficiently answer \LCEW queries when one of the arguments is a selected position.
For each selected position $i$ and each (arbitrary) position $j$, this table stores the distance from $i$ to the last (rightmost) selected position that appears in the longest common extension on the side of $i$, i.e., the last selected position $i'$ for which $T[i\dd i']$ matches $T[j\dd j+i'-i]$.
Formally,
\[\forall i\in\Ss, j\in[1\dd n] :
\jump[i, j] = \max\{i'-i : i'\in\Ss \land i'\ge i \land  T[i\dd i'] \matches T[j\dd j+i'-i]\}.\]
If there is no such selected position $i'$
(which happens when $T[i] \mismatches T[j]$), we set $\jump[i,j] := -\infty$.
This table contains $\lambda \cdot n = \cO(nG/t)$ entries
and allows us to jump from the first to the last selected position in the common extension, thus reducing \LCEW queries to finding longest common extensions \textit{to} and \textit{from} a selected position.

Finally, let $T_\#$ be the string obtained by replacing all wildcards in $T$ with a new character~``$\#$'' that does not appear in $T$.
The string $T_\#$ does not contain wildcards, and for any $i,j\in[1\dd n]$, we have $\LCEW_{T}(i,j) \ge \LCE_{T_\#}(i,j)$.

\subparagraph*{The data structure.} Our data structure consists of
\begin{itemize}
   \item the \jump table,
   \item the arrays \nextpos and \nextsel, and
   \item a data structure for constant-time \LCE queries in $T_\#$, with $\cO(n)$ construction time and $\cO(n)$ space usage (e.g., a suffix tree augmented with a lowest common ancestors data structure~\cite{10.1007/11780441_5}).
 \end{itemize} 

The \jump table uses space $\cO(nG/t)$ and can be computed in time $\cO(n(G/t)\log n)$ (see \cref{sec:compute-jump}), while the \nextpos and \nextsel arrays 
can be computed in $\cO(n)$ time and stored using $\cO(n)$ space.
Therefore, our data structure can be built in $\cO(n + n(G/t)\cdot \log n) = \cO(n(G/t)\cdot \log n)$ time and requires $\cO(n+nG/t) = \cO(nG/t)$ space.
As shown in \cref{sec:queries}, we can use this data structure to answer \LCEW queries in $T$ in time $\cO(t)$, thus proving \cref{thm:data-structure}.

\subsection{Computing the \jump Table}\label{sec:compute-jump}

In this section, we prove the following lemma.
\begin{lemma}\label{lemma:compute-table}
  Given random access to $T$,
  the \jump table can be computed in $\cO(n(G/t)\cdot \log n)$ time
  and $\cO(nG/t)$ space.
\end{lemma}
\begin{proof}
To compute the \jump table, we leverage the algorithm of Clifford and Clifford for exact pattern matching with wildcards~\cite{clifford2007simple}.
This algorithm runs in time $\cO(n \log m)$, and finds all occurrences of a
pattern of length $m$ within a text of length $n$ (both may contain wildcards).

Let $i_1 < i_2 < \ldots < i_\lambda = n$ denote the \textit{selected} positions,
sorted in the increasing order.
For $r = 1,\ldots,{\lambda-1}$, let $P_r$ be the fragment of $T$
from the $r$-th to the $(r+1)$-th selected position (exclusive),
i.e.,  $P_r = T[i_r\dd i_{r+1})$, and let $\ell_r$ denote the length of $P_r$,
that is $\ell_r = |P_r| = i_{r+1} - i_r$.
Then, for every $r$, we use the aforementioned algorithm of Clifford and Clifford~\cite{clifford2007simple} to compute the occurrences of $P_r$ in $T$: it returns an array $A_r$ such that $A_r[i] = 1$ if and only if $T[i\dd i+\ell_r) \matches P_r$.

Using the arrays $(A_r)_r$, the \jump table can then be computed with a dynamic programming approach, in the spirit of the computations in~\cite{crochemore2015note}.
The base case is $i = i_\lambda$, for which we have, for all $j\in[1\dd n]$, $\jump[i_\lambda, j] = 0$ if $T[i_\lambda] \matches T[j]$ and $-\infty$ otherwise. 
We can then fill the table by iterating
over all pairs $(r,j)\in [1\dd \lambda-1]\times[1\dd n]$ in the reverse lexicographical order
and using the following recurrence relation:
\begin{equation}\label{eq:jump}
\jump[i_r, j] = 
\begin{cases}
  -\infty &\text{ if } T[i_r] \mismatches T[j]\\
  \max(0, \ell_r + \jump[i_{r+1}, j+\ell_r])& \text{if } T[i_r] \matches T[j], A_r[j] = 1\\
  0 & \text{ otherwise.}\\
\end{cases}
\end{equation}

Computing the arrays $(A_r)_r$
takes $\cO(\lambda \cdot n\log n) = \cO(n (G/t) \cdot \log n)$ time in total.
Computing the $\jump$ table from the arrays takes constant time per cell,
and the table contains $\lambda \cdot n = \cO(nG/t)$ cells.
Thus, the $\jump$ table can be computed in time $\cO(n (G/t) \cdot \log n)$.
\end{proof}

\subsection{Answering \LCEW Queries}\label{sec:queries}
Our algorithm for answering \LCEW queries can be decomposed into the following steps:
\begin{enumerate}[(a)]
  \item \label{step:start} move forward in $T$ until we reach a selected position or a mismatch,
  \item \label{step:jump} use the \jump table to skip to the last selected position in the longest common prefix on the side of the selected position,
  \item \label{step:end} move forward until we either reach a mismatch or the end of the text.
\end{enumerate}
Steps (\ref{step:start}) and (\ref{step:jump}) might have to be performed twice, one for each of the ``sides'' of the query. 
Steps (\ref{step:start}) and (\ref{step:end})
can be handled similarly,
using \LCE queries in $T_\#$ to move forward multiple positions at a time:
see \cref{alg:subroutine} for a pseudo-code implementation of these steps.
We provide a pseudo-code implementation of the query procedure as \cref{alg:query}.

\begin{center}
\begin{minipage}{0.85\textwidth}
    \begin{algorithm}[H]
\caption{Subroutine for $\LCEW$ queries}
\label{alg:subroutine}
\begin{algorithmic}[1]
\Function{NextSelectedOrMismatch}{$i, j$}
  \State $\ell \gets 0$
  \State $m\gets \min(\nextsel[i], \nextsel[j])$
  \While{$T[i+\ell] \matches T[j+\ell]$ and $i+\ell\notin\Ss$ and $j+\ell\notin\Ss$}
    \State $r \gets \LCE_{T_\#}(i+\ell, j+\ell)$\label{line:lce}
    \State $\ell \gets \min(\ell+r, m)$\label{line:min1}
    \State $d \gets 0$
    \If{$T[i+\ell] =\wild$}
      \State $d \gets \max(d, \nextpos[i+\ell])$
    \EndIf
    \If{$T[j+\ell] =\wild$}
      \State $d \gets \max(d, \nextpos[j+\ell])$
    \EndIf
    \State $\ell\gets \min(\ell+d, m)$\label{line:min2}
  \EndWhile
  \State \Return $\ell$
\EndFunction
\end{algorithmic}
\end{algorithm}
\end{minipage}
\end{center}

\subparagraph*{Analysis of the \textsc{NextSelectedOrMismatch} subroutine (\cref{alg:subroutine}).}
\cref{alg:subroutine} computes a value $\ell$
such that $T[i\dd i+\ell) \matches T[j\dd j+\ell)$,
and either $T[i +\ell] \mismatches T[j +\ell]$ or at least one of $i+\ell, j+\ell$ is a selected position. In the latter case, $i+\ell$ (resp.~$j+\ell$) is the first selected position after $i$ (resp.~$j$).
Furthermore, \cref{alg:subroutine} runs in time $\cO(t)$.
These properties are formally stated and proved below.

\begin{restatable}{lemma}{mismatchorselected}
\label{claim:mismatch-or-selected}
  Let $\ell$ be the value returned by \cref{alg:subroutine}.
  We have $T[i\dd i+\ell) \matches T[j\dd j+\ell)$.
\end{restatable}
\begin{proof}
  We prove that $T[i\dd i+\ell) \matches T[j\dd j+\ell)$
  by induction on the number of iterations of the \texttt{while} loop.
  At the start of the algorithm, we have $\ell = 0$, $i = i+\ell$ and $j=j+\ell$,
  hence the base case holds.

  Now, assume that at the start of some iteration of the \texttt{while} loop,
  we have $T[i\dd i+\ell) \matches T[j\dd j+\ell)$.
  In Line~\ref{line:lce}, $r =\LCE_{T_\#}(i+\ell, j+\ell)$,
  and hence $T_\#[i+\ell\dd i+\ell+r)$ is equal to $T_\#[j+\ell\dd j+\ell+r)$,
  and \emph{a fortiori}, the same thing is true in $T$,
  i.e., $T[i+\ell\dd i+\ell+r) \matches T[j+\ell\dd j+\ell+r)$.
  Combining the above with our invariant hypothesis,
  we obtain that $T[i\dd i+\ell+r) \matches T[j\dd j+\ell+r)$.
  Therefore, after Line~\ref{line:min1} is executed, we have $T[i\dd i+\ell) \matches T[j\dd j+\ell)$
  for the new value of $\ell$.
  By \cref{fact:match-to-nextpos}, at least one of $T[i+\ell\dd i+\ell+d)$ or $T[j+\ell\dd j+\ell+d)$
  consists only of wildcards, therefore these two fragments match,
  and, before executing Line~\ref{line:min2}, we have $T[i\dd i+\ell+d) \matches T[j\dd j+\ell+d)$.
  Finally, after executing Line~\ref{line:min2}, the above becomes $T[i\dd i+\ell) \matches T[j\dd j+\ell)$,
  and our induction hypothesis holds.
\end{proof}

The fact that either $T[i +\ell] \mismatches T[j +\ell]$ or at least one of $i+\ell, j+\ell$ is a selected position follows from the exit condition of the \texttt{while} loop.
If one of them is a selected position, the minimality of its index
follows from using $m = \min(\nextsel[i], \nextsel[j])$ to bound the value of $\ell$ throughout the algorithm.
This concludes the proof of the correctness of \cref{alg:subroutine}.

We next analyse the running time of \cref{alg:subroutine}.

\begin{lemma}\label{lem:nsorm}
After an $\cO(n)$-time preprocessing of $T$, \cref{alg:subroutine} runs in time $\cO(t)$.
\end{lemma}
\begin{proof}
We preprocess $T$ in $\cO(n)$ time so that $\LCE$ queries on $T_\#$ can be answered in $\cO(1)$ time~\cite{10.1007/11780441_5}.

It then suffices to bound the number of iterations of the \texttt{while} loop.
We do so using the following two claims.

\begin{claim}
\label{claim:progress}
  In Line~\ref{line:min1} of \cref{alg:subroutine}, 
  either $T[i +\ell] \mismatches T[j +\ell]$,
  or (at least) one of $T[i +\ell], T[j +\ell]$ is a selected position or a wildcard.
\end{claim}
\begin{proof}
  In Line~\ref{line:lce}, $r$ is the \LCE of $T_\#[i +\ell\dd n]$ and $T_\#[j +\ell\dd n]$,
  and therefore $T_\#[i +\ell+r] \neq T_\#[j +\ell+r]$.
  Then, in Line~\ref{line:min1} the value of $\ell$
  is set to either $\ell+r$ or $m$.

  In the former case, we either have $T[i +\ell] \mismatches T[j +\ell]$ (and we are done) or  $T[i +\ell] \matches T[j +\ell]$, and one of $T[i +\ell], T[j +\ell]$ is a wildcard
  as $T_\#[i +\ell+r] \neq T_\#[j +\ell+r]$.

  In the latter case, i.e., if $\ell$ is set to $m = \min(\nextsel[i], \nextsel[j])$
  in Line~\ref{line:min1}, then, by the definition of $\nextsel$,
  at least one of $T[i +\ell], T[j +\ell]$ is a selected position.
\end{proof}

\begin{claim}
\label{claim:transition}
  In Line~\ref{line:min2} of \cref{alg:subroutine}, 
  either $T[i +\ell] \mismatches T[j +\ell]$,
  or (at least) one of $T[i +\ell], T[j +\ell]$ is a transition position.
\end{claim}
\begin{proof}
  By \cref{claim:progress}, we have that in Line~\ref{line:min1} of \cref{alg:subroutine}, 
  either $T[i +\ell] \mismatches T[j +\ell]$,
  or at least one of $T[i +\ell], T[j +\ell]$ is a selected position or a wildcard.
  We consider three sub-cases.

  If $T[i +\ell] \mismatches T[j +\ell]$ in Line~\ref{line:min1},
  then neither $T[i +\ell]$ nor $T[j +\ell]$ 
  is a wildcard, and $d$ is $0$ in Line~\ref{line:min2}. Hence, the value of $\ell$ does not change.

  If one of $T[i +\ell], T[j +\ell]$ is a selected position in Line~\ref{line:min1},
  then we have $\ell = m$ by the minimality of $m$,
  and $\ell$ will be set to the same value in line~\ref{line:min2} regardless of the value of $d$.

  Finally, assume that one of $T[i +\ell], T[j +\ell]$ is a wildcard in Line~\ref{line:min1}.
  Then for any $p\in\{i +\ell, j +\ell\}$ such that $T[p]$ is a wildcard,
  $T[p + \nextpos[p]]$ is a transition position (by the definition of $\nextpos$).
  Before executing Line~\ref{line:min2}, $d$ is the maximum of these $\nextpos[p]$, and
  hence one of $T[i +\ell + d]$ or $T[j +\ell + d]$ is a transition position.
\end{proof}

By \cref{claim:transition}, the number of transition positions between $i+\ell$ or $j+\ell$
and the corresponding next selected position decreases by at least one
(or the algorithm exits the loop and returns).
The use of $m$ in Lines~\ref{line:min1} and~\ref{line:min2}
ensures that we cannot go over a selected position,
and, by construction, there are at most $t$ transition positions between $i$ or $j$ and the next selected position,
therefore \cref{alg:subroutine} goes through at most $2t$ iterations of the loop.
Each iteration consists of one \LCE query in $T_\#$ and a constant number of constant-time operations.
Hence, \cref{alg:subroutine} takes time $\cO(t)$ overall.
\end{proof}

\subparagraph*{\LCEW query algorithm.}
Let $\ell$ denote the result of \cref{alg:query} on some input $(i,j)$.
The properties of \cref{alg:subroutine} ensure that $T[i\dd i+\ell) \matches T[j\dd j+\ell)$.
As the algorithm returns when it either encounters a mismatch or reaches the end of the string,
the matching fragment cannot be extended, which ensures the maximality of $\ell$.
To prove that \cref{alg:query} runs in time $\cO(t)$,
we show that it makes a constant number of loop iterations.

\begin{center}
\begin{minipage}{0.85\textwidth}
\begin{algorithm}[H]
\caption{Algorithm to answer the query $\LCEW(i, j)$}
\label{alg:query}
\begin{algorithmic}[1]
\Function{LCEW}{$i,j$}
  \State $\ell \gets 0$
  \While{$i+\ell \le n$ and $j+\ell \le n$}
    \State $\ell \gets \Call{NextSelectedOrMismatch}{i+\ell,j+\ell}$\label{line:subroutine}
    \If{$T[i+\ell] \mismatches T[j+\ell]$}
      \State \Return $\ell$\label{line:mismatch}
    \EndIf
    \If{$i+\ell\in\Ss$}
      \State $\ell \gets \ell + \jump[i+\ell, j+\ell]+1$\label{line:base1}
    \Else
      \State $\ell \gets \ell + \jump[j+\ell, i+\ell]+1$\label{line:base2}
    \EndIf
  \EndWhile
  \State \Return $\ell$
\EndFunction
\end{algorithmic}
\end{algorithm}
\end{minipage}
\end{center}

\begin{restatable}{lemma}{threeiter}
\label{claim:three-iter}
  The \texttt{while} loop of \cref{alg:query} makes at most three iterations.
\end{restatable}
\begin{proof}
  After the call to \cref{alg:subroutine} in Line~\ref{line:subroutine},
  either $T[i +\ell] \mismatches T[j +\ell]$ or at least one of $i+\ell, j+\ell$ is a selected position.
  In the former case, this is the last iteration of the loop.
  In the latter case, suppose without loss of generality that $i+\ell$ is a selected position.
  Then, in Line~\ref{line:base1}, the value of $\ell$ is updated to $\jump[i+\ell, j+\ell]+1$,
  and there is no selected position between $i+\ell$ and the end of the longest common extension with wildcards.
  This can happen at most once for each of $i$ and $j$,
  and thus the loop goes through at most three iterations before exiting. 
\end{proof}

Therefore, \cref{alg:query} makes up to three calls to \cref{alg:subroutine} plus a constant number of operations, and thus runs in time $\cO(t)$ due to \cref{lem:nsorm}.

\section{Faster Approximate Pattern Matching and Computation of Periodicity Arrays construction}\label{sec:applications}
In this section, we use the data structure of \cref{thm:data-structure}
to derive improved algorithms for the \kpmwe problem and the problem of computing periodicity arrays of strings with wildcards.

\subsection{Faster Pattern Matching with Errors and Wildcards}\label{sec:pattern-matching}
We first consider the problem of pattern matching with errors,
where both the pattern and the text may contain wildcards. The edit distance between two strings $X,Y \in \Sigma^\ast$, denoted by $\ed(X,Y)$, is the smallest number of insertions, deletions, and substitutions of a character in $\Sigma$ by another character in $\Sigma$, required to transform $X$ into a string matching $Y$. This problem is formally defined as follows:

\problemoutput{$k$-PMWE}%
{A text $T$ of length $n$, a pattern $P$ of length $m$ and an integer threshold $k$.}%
{Every position $p$ for which there exists $i\le p$ such that $\ed(T[i\dd p], P) \le k$.}

Akutsu~\cite{akutsu} gave an algorithm for this problem
that runs in time $\cOtilde(n\sqrt{km})$.
Using their framework, we show that the complexity can be reduced to $\cO(n(k+\sqrt{kG \log m}))$, where $G$ is 
the cumulative number of groups of wildcards in $P$ and $T$
(or equivalently, the number of groups of wildcards in $P\$T$).\footnote{The additive ``$k$'' term in our complexity is necessary
because $G\log m$ might be smaller than $k$. On the other hand,
one can assume w.l.o.g.~that $m\ge k$ and hence, this additional term is hidden
in the running time of Akutsu's algorithm~\cite{akutsu}.}

\begin{theorem}\label{thm:alg-pm}
  There is an algorithm for \kpmwe that runs in time $\cO(n (k + \sqrt{kG \log m}))$.
\end{theorem}
\begin{proof}
  Akutsu~\cite[Proposition 1]{akutsu} 
  shows that, if after an $\alpha$-time preprocessing, 
  \LCEW queries between $P$ and $T$ can be answered in time $\beta \ge 1$,
  then the \kpmwe problem can be solved in time $\cO(\alpha + n\beta k)$.
  First, assume that $G\log m \ge k$.
  We use the data structure of \cref{thm:data-structure}
  with $t = \sqrt{(G/k) \cdot \log m} \ge 1$ to answer \LCEW queries:
  we then have $\alpha = \cO(n \sqrt{Gk \log m})$ (here, we use the \emph{standard trick} to replace the $\log n$ factor in the construction time with $\log m$ if $n \ge 2m$, by considering $\cO(n/m)$ fragments of $T$ of the form $T[i\cdot m + 1 \dd \lceil(i+3/2)\cdot m\rceil +k]$ and building an $\LCEW$ data structure for the concatenation of $P$ and each such fragment independently) 
  and $\beta = \cO(t) = \cO(\sqrt{(G/k)\cdot\log m})$.
  Therefore, the running time of the algorithm is $\cO(n \sqrt{Gk \log m})$.
  Second, if $G\log m < k$, we simply set $t = 1$:
  the total running time is then $\cO(nG \log m + nk) = \cO(nk)$.
  Accounting for both cases, the time complexity of this algorithm
  is $\cO(n(k+\sqrt{G k \log m}))$.\qedhere
\end{proof}

\subsection{Faster Computation of Periodicity Arrays}\label{sec:per-arrays}
Our data structure also enables us to obtain efficient algorithms for computing periodicity arrays of a string with wildcards (\cref{thm:arrays}).
These algorithms build on and improve upon the results of
Iliopoulos and Radoszewski~\cite{iliopoulos2016truly}.

\begin{theorem}\label{thm:arrays}
  Let $S$ be a string of length $n$ with $G$ groups of wildcards.
  The prefix array, the quantum and deterministic border arrays
  and the quantum and deterministic period arrays
  of $S$ can be computed in $\cO(n\sqrt{G}\log n)$ time and $\cO(n)$ space. 
\end{theorem}

By \cref{fact:pref-border}, it remains to show that the prefix array of $S$ can be computed
in $\cO(n\sqrt{G}\log n)$ time and $\cO(n)$ space. 
Recall that the \emph{prefix array} of a string $S$ of length $n$ is an array $\pi$ of size $n$ such that $\pi[i] = \LCEW(1, i)$. Consequently, $\pi$ can be computed using $n$ \LCEW queries in~$S$.
By instantiating our data structure with $t = \sqrt{G}$,
we obtain an algorithm running in $\cO(n\sqrt{G}\log n)$ time,
but its space usage is $\Theta(n\sqrt{G})$.
Below, we show how one can slightly modify the data structure of \cref{thm:data-structure} to reduce the space complexity to $\cO(n)$,
extending the ideas of~\cite{iliopoulos2016truly}.
\begin{lemma}
  Let $S$ be a string of length $n$ with $G$ groups of wildcards.
  The prefix array of~$S$ can be computed in $\cO(n\sqrt{G}\log n)$ time and $\cO(n)$ space. 
\end{lemma}
\begin{proof}
We add the index $1$ to the set of selected positions $\Ss$ and preprocess $S$ in $\cO(n)$ time and space to support $\LCE$ queries on $S_\#$ in $\cO(1)$ time~\cite{10.1007/11780441_5}. 

Notice that, using the dynamic programming algorithm of \cref{lemma:compute-table}, for any $r< \lambda$, the row $(\jump[i_r, j], j=1,\ldots,n)$ of the \jump table can be computed in $\cO(n\log n)$ time and $\cO(n)$ space from the next row $(\jump[i_{r+1}, j], j=1,\ldots,n)$. It suffices to compute the array~$A_r$ of occurrences of $P_r$ in time $\cO(n\log n)$ using the algorithm of Clifford and Clifford~\cite{clifford2007simple}, and then apply the recurrence relation of \cref{eq:jump}.
This way, we can compute the first row of the \jump table in $\cO(n (G/t) \cdot \log n)$ time using $\cO(n)$ space.

To answer an $\LCEW(1, j)$ query, we perform the following steps: first, we use $\jump[1, j]$ to jump to the last selected position of the longest common extension on the ``side of $1$'' and then perform at most $t$ regular $\LCE$ queries as in \cref{alg:subroutine}.
If after this process we reach a mismatch or the end of $S$, we are done. Otherwise, we need to perform another $\jump$ query from indices $(\ell_j, i_{r_j})$, where $i_{r_j} = j+\ell_j-1$ is a selected position, and then perform at most $t$ more regular \LCE queries from the resulting positions. We store the indices $(\ell_j, i_{r_j})$ for every value~$j$, grouped by selected position $i_{r_j}$.

To answer all required queries of the form $\jump[\ell_j, i_{r_j}]$ queries in $\cO(n (G/t) \cdot \log n)$ time and $\cO(n)$ space, we again recompute the \jump table row by row, starting from row $\lambda$ and going up, storing only one row at a time.
After computing the row corresponding to a selected position~$i_r$, we answer all \jump queries with $i_{r_j} = i_r$ and then perform the remaining \LCE queries to answer the underlying $\LCEW(1, j)$ query.

The claimed bound follows by setting $t := \sqrt G \leq G$.
\end{proof}

\section{Connection to Boolean Matrix Multiplication}
In this section, we describe a fine-grained connection between \LCEW
and (sparse) Boolean matrix multiplication.
In \cref{sec:lower-bound}, we use this connection to obtain a lower bound
on the preprocessing-query-time product of combinatorial data structures for \LCEW.
In \cref{sec:sparse}, we further connect sparse matrices and strings with few groups
of wildcards, deriving an efficient multiplication algorithm. 
Both results are based on the following lemma.

\begin{lemma}\label{lem:reduce_bmm_to_lce}
    Let $A$ and $B$ be Boolean matrices of respective dimensions $a \times b$ and $b \times c$. In $\cO(ab + bc)$ time, one can compute a string $S$ of length $ab + bc$ over alphabet $\{\wild, 0, 1\}$ such that each cell of $AB$ can be computed with an \LCEW query in $S$.
    The query indices can be obtained from the row and column indices of the cell in $\cO(1)$ time.
    The number of groups of wildcards in $S$ is at most $\min(m_A, m_B) + 1$, where $m_A$ and $m_B$ are the respective number of $1$-bits in $A$ and~$B$.
\end{lemma}

\begin{proof}
    For now, assume $m_A \leq m_B$. We encode $A$ into a string $S_A$ of length $ab$ in row-major order, i.e., $S_A[b \cdot i + j + 1] = \phi_A(A[i+1, j+1])$ for $i \in [0\dd a),j \in[0\dd b)$,
  and $B$ into a string $S_B$ of length $bc$ in column-major order,
  i.e., $S_B[i + b \cdot j + 1] = \phi_B(B[i+1, j+1])$ for $i \in [0\dd b),j \in[0\dd c)$,
  where
  \[
    \phi_A(x) =
      \begin{cases}
        1& \text{ if } x = 1,\\
        \wild& \text{ if } x = 0,
      \end{cases}
    \text{\qquad and\qquad}
    \phi_B(y) =
      \begin{cases}
        0& \text{ if } y = 1,\\
        1& \text{ if } y = 0.  
      \end{cases}
  \]
  It follows from this definition that $\phi_A(x)$ does not match $\phi_B(y)$, i.e., $\phi_A(x) \mismatches \phi_B(y)$, 
  if and only if $x = y = 1$. Since $S_A$ contains only $m_A$ non-wildcard symbols, it is easy to see that there are at most $m_A + 1$ groups of wildcards. We claim that each entry of $C = A\cdot B$ can be computed with one \LCEW query to $S$.
  
  \begin{claim}\label{claim:inner-product} For $i \in [0\dd b)$ and $j \in[0\dd c)$, it holds  %
  \[C[i+1,j+1] = 1 \Longleftrightarrow \LCEW_{S_A, S_B}(b\cdot i + 1, b\cdot j + 1) < b.\] \end{claim}
  \begin{claimproof} By the definition of $\LCEW$, it holds $\LCEW_{S_A, S_B}(b\cdot i + 1, b\cdot j + 1) < b$ if and only if there exists an index $k$ with $0 \leq k < b$ such that $S_A[b\cdot i + 1 + k] \mismatches S_B[b\cdot j + 1 + k]$.
  By construction, it holds $S_A[b\cdot i + k + 1] \mismatches S_B[k + b\cdot j + 1]$ if and only if $A[i+1,k+1] = B[k+1,j+1] = 1$. Finally, $C[i+1,j+1] = 1$ if and only if there exists an index $k$ such that $0 \leq k < b$ and $A[i+1,k+1] = B[k+1,j+1] = 1$, which concludes the proof of the claim.
  \end{claimproof}

  We assumed $m_A \leq m_B$. If $m_B < m_A$, then we simply swap the roles of $\phi_A$ and $\phi_B$. Evaluating $m_B < m_A$ and constructing the string clearly takes $\cO(ab + bc)$ time.
\end{proof}

\subsection{A Lower Bound for Combinatorial Data Structures}\label{sec:lower-bound}
Our lower bound is based on the combinatorial matrix multiplication conjecture which states that for any $\varepsilon > 0$ there is no combinatorial algorithm for multiplying two $n \times n$ Boolean matrices working in time $\cO(n^{3-\varepsilon})$. Gawrychowski and Uzna\'{n}ski~\cite[Conjecture 3.1]{DBLP:conf/icalp/GawrychowskiU18} showed that the following formulation is equivalent to this conjecture:

\begin{conjecture}[Combinatorial matrix multiplication conjecture]\label{conj:matrix-mul}
  For every $\eps > 0$ and every $\alpha, \beta, \gamma > 0$,
  there is no combinatorial algorithm that computes the product
  of Boolean matrices of dimensions $n^\alpha \times n^{\beta}$ and $n^{\beta}\times n^{\gamma}$
  in time $\cO(n^{\alpha+\beta+\gamma-\eps})$. 
\end{conjecture}

\begingroup
\def\nu{g}
\def\mu{q}
\begin{theorem}\label{thm:lower-bound}
    Consider constants $\eps, \mu, \nu \in \mathbb R^+$ with $0 < \mu < \nu < 1$. 
    Consider a combinatorial data structure for \LCEW on a string of length $\Theta(n)$ that contains $\cO(n^\nu)$ wildcards with query time $\cO(n^\mu)$.
    Under \cref{conj:matrix-mul}, the preprocessing time cannot be $\cO(n^{1 + \nu - \mu - \eps})$.
\end{theorem}

\begin{proof}
  Since the claim is stronger for smaller $\eps$, we can assume $\eps/2 < \nu - \mu$ without loss of generality.
  We define positive constants $\alpha = \nu - \mu - \eps/2$, $\beta = \mu + \eps/2$, and $\gamma = 1 - \mu - \eps/2$.
  Let $A$ and $B$ be Boolean matrices of respective dimensions $a\times b$ and $b \times c$, where $a = \ceil{n^{\alpha}}$, $b = \ceil{n^\beta}$, and $c = \ceil{n^\gamma}$.
  Our goal is to compute $AB$ of dimensions $a \times c$. 

  We use \cref{lem:reduce_bmm_to_lce} to construct a string $S$ of length $ab + bc = \Theta(n)$ that contains at most $ab = \cO(n^\nu)$ wildcards. Now each of the $ac = \Theta(n^{1 + \nu - 2\mu - \eps})$ cells of $AB$ can be computed with one \LCEW query. Hence, if there is an \LCEW data structure for $S$ that has query time $\cO(n^\mu)$ and preprocessing time $\cO(n^{1 + \nu - \mu - \eps})$, then the overall time for computing $AB$ is $\cO(n^{1 + \nu - \mu - \eps}) = \cO(n^{\alpha + \beta + \gamma - \eps/2})$, which contradicts \cref{conj:matrix-mul}.
\end{proof}
\endgroup

\subsection{Fast Sparse Matrix Multiplication}\label{sec:sparse}
In this section, we further connect sparse matrices and strings with few groups
of wildcards, deriving an algorithm for sparse Boolean matrix multiplication (BMM). 
\cref{lem:reduce_bmm_to_lce} already relates wildcards and the number of $1$-bits in the input matrices.
Now we also consider the number $m_{out}$ of $1$-bits in the output matrix.
While the reduction underlying the proof of \cref{thm:lower-bound} uses one
\LCEW query for each entry of the output matrix, we next show that, using the same reduction, we can actually compute multiple $0$-bits of the output matrix with a single \LCEW query. This will ultimately allow us to use less than $a+c+m_{out}$ queries.

\begin{lemma}\label{lem:reduce_bmm_to_lce:diagonals}
    Let $A, B$ be Boolean matrices of respective dimensions $a \times b$ and $b \times c$.
    Let $i \in [0\dd a)$, $j \in [0\dd c)$, and $x \in [1..\min(a - i, c - j)]$. For the reduction from \cref{lem:reduce_bmm_to_lce}, it holds $\LCEW_{S_A, S_B}(bi + 1, bj + 1) \geq bx$ if and only if $\forall y \in [1\dd x] : (AB)[i + y, j + y] = 0$.
\end{lemma}
\begin{proof}
  First, if $\LCEW(bi + 1, bj + 1) \ge  bx$,
  then for every $y \in [1\dd x]$ it holds \[\LCEW(b(i+y-1) + 1, b(j+y-1) + 1) = \LCEW(bi + 1, bj + 1) - by + b \ge bx - by + b \ge b.\]
  By \cref{claim:inner-product}, this implies that $(AB)[i+y,j+y] = 0$.
  For the converse, we exploit that $\LCEW(i', j') \geq b$ and $\LCEW(i' + b, j' + b) \geq \ell$
  implies $\LCEW(i',j') \ge b + \ell$ for any query positions $i',j'$ and integer $\ell$. If $(AB)[i + y, j + y] = 0$ for $y \in [1\dd x]$, then by \cref{claim:inner-product} also $\LCEW(b(i+y-1) + 1, b(j+y-1) + 1) \geq b$, and it readily follows $\LCEW(bi + 1, bj + 1) \geq bx$.
\end{proof}

The above lemma implies that the answer to an \LCEW query at the first indices of a diagonal gives us
the length of the longest prefix run of zeroes in this diagonal.
A repeated application of this argument implies that
computing the entries in the $d$-th diagonal of $C$ takes $m_{d} +1$ \LCEW queries,
where $m_{d}$ is the number of non-zero entries in this diagonal.
Summing over all $a+c-1$ diagonals, this gives a total of $a + c+m_{out} - 1$ queries. 
Now we use this insight to derive an algorithm for the multiplication of sparse matrices.

\begin{theorem}\label{thm:sparseMM_rec}
  Let $A, B$ be Boolean matrices of respective dimensions $a \times b$ and $b \times c$, and let $C = AB$. Let $m_A$, $m_B$, and $m_C$ denote the respective number of $1$-bits in $A$, $B$, and $C$.
  There is a deterministic combinatorial algorithm for computing $C$ that runs in time
  $\cO((m_A + m_B) \cdot \log (m_A + m_B) + \sqrt{(ab + bc) \cdot \min(m_{A}, m_B) \cdot(a + c + m_{C}) \cdot \log (ab + bc)})$.
\end{theorem}\begin{proof}

\def\emin{m_{\textnormal{in}}}
\def\emout{m_{\textnormal{out}}}

  We assume that the matrices are given as a list of coordinates of $1$-bits, sorted by row index and then by column index.
  This compact representation has size $\cO(m_{A} + m_{B})$.

  Let $m_{in} = \min(m_A, m_B)$. We assume without loss of generality that $m_{in} \geq b$. Otherwise, there must be empty (i.e., all-zero) columns in $A$ or empty rows in $B$.
  If either the $i$-th column of $A$ or the $i$-th row of $B$ is empty (or both), then we can simply remove both the $i$-th row of $A$ and the $i$-th column of $B$ without affecting the result of the multiplication.
  Eliminating empty rows and columns in this way can be achieved by appropriately offsetting the indices of non-zeroes in the list representation, which can be easily done in $\cO((m_A + m_B) \cdot \log (m_A + m_B))$ time. Afterwards, there are respectively at most $m_{in}$ columns in $A$ and rows in $B$.

  Consider the string $S$ from the reduction in \cref{lem:reduce_bmm_to_lce}, which has length $ab + bc$ and
  contains $G \leq m_{in} + 1$ groups of wildcards, and let
  \[t = \Ceil{\sqrt{(ab + bc) \cdot G \cdot \log(ab + bc) / (a + c + m_C)}}.\]
  If $t \le G$, then we build the data structure from \cref{thm:data-structure} for $S$ with this parameter~$t$.
  Otherwise, that is, if $t > G$, we use the data structure from \cref{cor:extended-tradeoff} instead, which is possible because $t < 1 + \sqrt{(ab + bc) \cdot G \cdot \log (ab + bc)} = 1 + \sqrt{\absolute{S} \cdot G \cdot \log \absolute{S}} \leq G \cdot |S|^{3/4}$ for $|S|, G > 10$.
  Then, using \cref{lem:reduce_bmm_to_lce:diagonals} and the argument described above this \lcnamecref{thm:sparseMM_rec}, we compute $C$ using less than $a+c+m_C$ queries. The total query time is $\cO({\sqrt{(ab + bc) \cdot m_{in} \cdot \log(ab + bc) \cdot (a + c + m_C)}})$.
  (This holds when $t = 1$ since $m_C \leq m_A \cdot m_B = \max(m_A, m_B) \cdot m_{in} \leq (ab + bc) \cdot m_{in}$.)
  In the case when we use the data structure encapsulated in \cref{thm:data-structure}, the preprocessing time matches the query time.
  In the remaining case when we use the data structure from \cref{cor:extended-tradeoff}, we need $\cO(ab + bc)$ preprocessing time. Due to $m_{in} \geq b$, the preprocessing time is still dominated by the time required for all queries, and hence the total time complexity is as claimed.

  Choosing $t$ requires knowing the value of $m_C$. If $m_C$ is unknown, we estimate it using exponential search, starting with estimate $a + c$ and doubling the estimate in every step. For a given estimate of $m_C$, we run the algorithm, halting and restarting whenever the total query time exceeds the construction time.
  Since the time budget increases by a factor of $\sqrt{2}$ in each round, the total time is dominated by the final round, in which we run the algorithm with a constant factor approximation of $m_C$. Hence the exponential search does not asymptotically increase the time complexity.
\end{proof}

The following corollary for square matrices is immediate.

\begin{corollary}\label{thm:sparseMM}
  Let $A, B$ be Boolean matrices, each of size $n \times n$, such that the total number of $1$-bits in $A$ and $B$ is $m_{in}$, and the number of $1$-bits in $AB$ is $m_{out}$.
  There is a deterministic combinatorial algorithm that computes $AB$ in
  $\cO(n\sqrt{m_{in} \cdot(n+ m_{out}) \log n})$ time.
\end{corollary}


\section{Conditional Lower Bounds from \tsum and \sd}

In this section, we show conditional lower bounds for the space vs query time trade-off of the \LCEW problem and for the preprocessing time vs query time trade-off of the \LCEW problem.
The lower bounds are conditional on the hardness of the \tsum
and \sd problems, which are defined below.

\problemtask{\tsum}%
{A set $S \subseteq [-n^3 \dd n^3]$ of size $n$.}%
{Decide if there exist distinct $x,y,z \in S$ with $x+y=z$.}

\problemquery{\sd}%
{Sets $S_1, S_2, \ldots , S_r$ of total size $N$.}%
{Given $i,j \in [1 \dd r]$, decide if $S_i \cap S_j = \emptyset$.}

It is conjectured that \tsum cannot be solved in strongly sub-quadratic time. We will obtain conditional lower bounds on the preprocessing time of \LCEW data structures by exploiting that \tsum can be reduced to \sd.

\begin{conjecture}[\tsumconj\ {\cite{DBLP:journals/comgeo/GajentaanO12,DBLP:conf/stoc/Patrascu10}}]\label{conj:tsum}
\tsum requires $n^{2-o(1)}$ time.
\end{conjecture}

\begin{theorem}[{\cite[Corollary 1.8]{DBLP:journals/corr/KopelowitzPP14}, see also~\cite{DBLP:conf/soda/KopelowitzPP16}}\protect\footnote{For the sake of uniform notation, we use a slightly different but equivalent formulation of the \lcnamecref{thm:3sum-sd}.}]\label{thm:3sum-sd}
For every constant $\eps > 0$ and $0 \leq q \leq \frac12$, there is no data structure for \mbox{\sd} with $\cO(N^{2 - 2q - \eps})$ preprocessing time and $\cO(N^q)$ query time under the \tsumconj.
\end{theorem}

For conditional lower bounds on the space complexity of \LCEW data structures, we instead use the following conjecture on \sd.

\begin{conjecture}[\sdconjstrong\ {\cite{DBLP:conf/wads/GoldsteinKLP17}}\protect\footnote{For the sake of uniform notation, we use a slightly different but weaker formulation of the conjecture.}]\label{conj:sd_strong}
For every constant $\eps > 0$ and $0 \leq q \leq \frac12$, there is no data structure for \sd occupying $\cO(N^{2 - 2q - \eps})$ space and answering queries in $\cO(N^{q})$ time.
\end{conjecture}

\subsection{Reduction from \sd to \LCEW}

The main workhorse of this section is a reduction from \sd to \LCEW; our technique is similar to a technique in a work of Kopelowitz and Vassilevska Williams \cite[Theorem~1]{DBLP:conf/icalp/KopelowitzW20}.
Note that the restrictions on the parameters in the \lcnamecref{lem:sdred:new} below are relatively weak; using the notation from the upper bound, we will use the \lcnamecref{lem:sdred:new} to derive conditional lower bounds for a wide range of $t$ and $G$, more precisely, whenever $t \leq n^{\frac13-\eps}$ and $G > t^{2}\cdot n^\eps$.

\begin{lemma}\label{lem:sdred:new}
Consider constants $q,s,g$ such that 
$0 \leq q \leq g \leq 1$ and $\max(\frac{q + 1}2, 2 + 2q - 3g) \leq s \leq 2 - q$.
Suppose that,
for strings of length $\cO(n)$ that contain $\cO(n^g)$ wildcards, there exists an $\LCEW$ data structure with query time $\cOtilde(n^q)$ and space (resp.\ preprocessing time) $\cOtilde(n^{s + g - 1})$.
Then, for a \sd instance of size $N$, there is a data structure with query time $\cOtilde(z^q)$ and space (resp.\ preprocessing time) $\cOtilde(z^{s})$, where $z = \Theta(N^{3/(s+ q + 1)})$.
\end{lemma}

\begin{proof}
We consider an instance of \sd over a collection $\mathcal S$ of sets $S_1, S_2, \ldots , S_r$ of total size $N$.
For the defined value $z$, it holds $N = \cO(z^s)$ because $3s \geq s + q + 1$ (which follows from $\frac{q + 1}2 \leq s$). Since $s$ and $q$ are constant, it holds $\log z = \Theta(\log N)$.
If one of the sets in a
query is of size at most $z^q$, we can search for its elements in the other set in time $\cO(z^q \log N)$ at query time -- assuming that we store each set as a balanced binary search tree.
These binary trees can be constructed in $\cO(N \log N) \subseteq \cOtilde(z^s)$ time and space.
We thus henceforth assume that all sets in $\mathcal{S}$ are of size at least~$z^q$. In particular, this implies that $r \leq N/z^q$.

We now want to remove some elements from our sets in order to reduce the size of the universe $\bigcup_{S\in \mathcal{S}} S$.
Let $y = 
\ceil{z^s / N}$.
We call each integer $x$ that appears in at most $y$ sets from~$\mathcal{S}$ \emph{infrequent}; the remaining integers are called \emph{frequent}.
Let $\mathcal{T}$ be an initially empty balanced binary search tree that will store triplets in the lexicographic order.
For each infrequent integer~$x$,
for each pair $(i,j) \in [1\dd r]^2$ with $i \neq j$ and $x \in S_i \cap S_j$, we insert the triplet $(i,j,x)$ into~$\mathcal{T}$.
In total, we insert $\cO(Ny) = \cO(z^s)$ triplets into~$\mathcal{T}$ because each $x \in S_i$ can be contained in $\cO(y)$ triplets.
This procedure can be implemented in $\cO(Ny \log N) \subseteq \cOtilde(z^s)$ time and space straightforwardly if we first sort the elements of the set $\bigcup_{i \in [1 \dd r]} \{(x,i) : x \in S_i\}$.
Now, given distinct $i, j \in [1 \dd r]$, we can check whether sets $S_i$ and $S_j$ contain an infrequent integer in their intersection in $\cO(\log N) \subset \cOtilde(z^q)$ time.

Observe that the total number of frequent integers is at most $N/y$.
We replace each set~$S_i$ by its intersection with the set of frequent integers in $\cO(N \log N)$ time in total.
After renaming those integers,
we can assume that the elements of all sets are from universe $[1 \dd u]$ with $u = \lfloor N / y \rfloor = \Theta(N^2 / z^s)$.
We henceforth work with these modified sets.

Now we can assume that there are $r = \cO(N/z^q)$ sets, each of which is a subset of $[1\dd u]$, which is of size $\Theta(N^2 / z^s)$.
We represent each set $S_i$ using two length-$u$ strings $L_i$ and $R_i$ over alphabet $\{0, 1, \wild \}$.
For $x \in [1\dd u]$, if $x \in S_i$ then $L_i[x] = 1$ and $R_i[x] = 0$. Otherwise, $L_i[x] = \wild$ and $R_i[x] = 1$.
Observe that the total size of all constructed strings is $2ru$, which is
\[\cO(N/z^q \cdot N^2/z^s) = \cO(N^3/ z^{s+q}) = \cO(N^{3 - 3(s+q)/(s+q+1)}) = \cO(N^{3/(s+q+1)}) = \cO(z).\]
Also, these strings can be constructed in $\cO(z)$ time.
For now, assume $g = 1$, i.e., we can afford an arbitrary number of wildcards.
We build the \LCEW data structure from the statement of the \lcnamecref{lem:sdred:new} for $T := L_1 L_2 \cdots L_r R_1 R_2 \dots R_r$.
For any $x \in [1\dd u]$ and distinct $i,j \in [1 \dd r]$, we have
$L_i[x]  \mismatches R_j[x]$
if and only if
$x \in S_i \cap S_j$.
Thus, $S_i \cap S_j = \emptyset$ if and only if $L_i$ matches~$R_j$, which can be evaluated in $\cO(z^q)$ time using a single \LCEW query in $T$.
The space (resp.\ preprocessing time) of the \LCEW data structure is $\cOtilde(z^s)$, as required.

Now we generalize the result to arbitrary $g$. Instead of a single string $T$, we use multiple strings $T_1,\dots, T_k$ (we specify $k$ later) such that each is of length $\cO(z)$ and contains $\cO(z^g)$ wildcards. Each string is responsible for answering a subset of queries. Recall that $u = \Theta(N^2 / z^s)$ and $2 + 2q - 3g \leq s$. The latter implies that $N^2 / z^s = z^{(2q + 2 - s)/3} \leq z^g$, and thus ${u = \cO(z^g)}$. Consequently, we can define a positive integer $h = \Theta(z^g / u)$.
The number of strings is $k = \ceil{\frac{r}{h}} = \cO(ru / z^g) = \cO(N/z^q) = \cO(z^{(s+q+1)/3} /z^g) = \cO(z^{1 - g})$ since $s+q+1 \leq 3$.
For $x \in [1\dd k]$, we define $T_x = L_{x'+ 1}L_{x' + 2}\dots L_{x' + h} \cdot R_1R_2\dots R_r$ with $x' = h \cdot (x - 1)$ (where handling the border case for $T_k$ is trivial).
Hence, each $L_i$ is contained in exactly one of the strings. 
A query $S_i, S_j$ can be answered using one \LCEW query in the string $T_x$ that contains~$L_i$.

It remains to be shown that we can construct \LCEW data structures with query time $\cO(z^q)$ for all the strings in overall  $\cO(z^s)$ space (resp.\ preprocessing time).
Every string is of length $\cO(z)$ and contains at most $u\cdot h = \cO(z^g)$ wildcards.
We construct $k = \cO(z^{1 - g})$ \LCEW data structures; each can be built in $\cOtilde(z^{s + g - 1})$ space (resp.\ preprocessing time). Hence the total space (resp.\ preprocessing time) for these data structures is $\cOtilde(z^s)$, as required. However, $T_1, \dots, T_k$ are of overall length $\cO(kz) = \cO(z^{2 - g})$, and we cannot afford to explicitly construct them if $s < 2 - g$. Instead, we only construct the single string $T$ of length $\cO(z)$. Given $T$, we can easily simulate constant time random access to any $T_x$.
\end{proof}

\subsection{Obtaining the Lower Bounds}

We next use the reduction from \cref{lem:sdred:new} to show conditional lower bounds on the trade-off of the query time against both the space and the preprocessing time. Using the notation of the upper bound, the result below states that significantly less than $\cO(nG/t^4)$ space and preprocessing time cannot be achieved for a wide range of $t$ and~$G$.

\begin{theorem}\label{thm:3sum_and_sd_conditional_lower}
    Consider constants $\eps, q, g$ such that $0 \leq q < \frac{1 - \eps}3$ and $2q + \frac{\eps}3 \leq g \leq 1$.
    For a string of length $\Theta(n)$ that contains $\cO(n^g)$ wildcards, consider an \LCEW data structure with $\cO(n^q)$ query time. Then both of the following statements hold.
    \begin{enumerate}[(i)]
        \vspace{.25\baselineskip}
        \item $\cO(n^{1 + g - 4q - \eps})$ preprocessing time cannot be achieved under the \tsumconj.
        \item $\cO(n^{1 + g - 4q - \eps})$ space cannot be achieved under
        the \sdconjstrong.
    \end{enumerate}
\end{theorem}

\begin{proof}
    For $(i)$, let $s = 2 - 4q - \eps < 2 - q$, and assume that the data structure can be constructed in $\cO(n^{1 + g - 4q - \eps}) = \cO(n^{s+g-1})$ time. 
    We have $6q + \eps \leq 3g$ and thus $2+2q-3g \leq 2-4q-\eps = s$. 
    Also, due to $q \leq \frac{1 - \eps}3$, we have $\frac{q + 1}2 \leq 2 - 4q - 1.5\eps < s$.
    Hence $q$, $g$, and $s$ satisfy the conditions required by \cref{lem:sdred:new}, and we obtain a data structure that solves a \sd instance of size $N$ in $\cO(z^s)$ preprocessing time and $\cO(z^q)$ query time, where $z = \Theta(N^{3/(s + q + 1)})$. Let $q' = 3q/(3 - 3q - \eps)$. For $0 \leq q < \frac{1-\eps}3$, and up to constant factors, we have
    \[z^q = N^{\frac{3q}{s + q + 1}} = N^{\frac{3q}{3-3q-\eps}} = N^{q'}\textnormal{, and}  \]
    \[z^s = N^{\frac{3s}{s + q + 1}} = N^{\frac{6-12q-3\eps}{3-3q-\eps}} = N^{\frac{6-6q-3\eps}{3-3q-\eps} - 2q'} =  N^{2 - 2q' - \frac{\eps}{3-3q-\eps}} \leq N^{2-2q'-{\eps}/3}.  \]
    It can be readily verified that $q' < \frac12$ (as $q'$ grows with $q$ and reaches its maximum when $q$ goes to $\frac{1 - \eps}3$). Hence, by \cref{thm:3sum-sd}, solving \sd with $\cO(z^s)$ preprocessing time and $\cO(z^q)$ query time contradicts the \tsumconj.
    
    The proof of $(ii)$ is almost identical. We initially assume that the \LCEW data structure uses $\cO(n^{s+g-1})$ space. Using \cref{lem:sdred:new}, we can solve a \sd instance of size $N$ in $\cO(z^s)$ space and $\cO(z^q)$ query time, contradicting the \sdconjstrong.
\end{proof}


\bibliographystyle{alphaurl}
\bibliography{references}
\end{document}